\documentclass{amsart}

\usepackage{amssymb}

\newtheorem{theorem}{Theorem}
\newtheorem{lemma}{Lemma}
\newtheorem{corollary}{Corollary}

\begin{document}

\title{Some C*-algebras associated to quantum gauge theories}

\author{Keith Hannabuss}
\address{Balliol College, Oxford OX1 3BJ, England.}

\date{29 July 2010 revised 16th October 2010}

\maketitle

\begin{abstract}
Algebras associated with Quantum Electrodynamics and other
gauge theories share some mathematical features with T-duality
Exploiting this different perspective and some category theory,
the full algebra of fermions and bosons can be regarded as a
braided Clifford algebra over a braided commutative boson
algebra, sharing much of the structure of ordinary Clifford
algebras.
\end{abstract}

\begin{center}
\it Dedicated to Alan Carey, on the occasion of his 60$\,^{th}$
birthday
\end{center}


\section*{Introduction}

It is just over 60 years since Quantum Electrodynamics achieved
its modern form, \cite{Sch,SSS}.  Some insights into its
ultraviolet and infrared divergences have been provided by
causal distribution splitting \cite{EG,Sc,FRS}, and the Hopf
algebra structure of nested Feynman diagrams \cite{K,CK},
respectively, both of which fit naturally within noncommutative
geometry \cite[Chapters 12-13]{GBVF}. With that in mind, we
present a slightly different perspective on the operator
algebras of gauge theories, which emphasizes noncommutative
geometric features, and also parallels some techniques which
appear in string theory and $T$-duality \cite{MR,BHM}.

For abelian gauge theories it is possible to give an explicit
derivation of their algebraic structures in terms of inducing
and crossed products. This parallels other examples in
noncommutative geometry, and the inducing process gives rise to
the Poisson--Gauss law relating the gauge field and fermion
charge density. These ideas are sketched in the next four
sections.

The remainder of the paper elucidates some features of the
algebraic construction more generally, using its functorial
properties. This reveals that the fermionic part of the theory
can be considered as a braided Clifford algebra over the gauge
bosonic algebra.


\section{Fermions}

The conventional operator approach to Quantum Electrodynamics
(QED) combines a fermionic anticommutation relation algebra
(CAR) and a bosonic commutation relation algebra (CCR). We
shall start with the fermions and introduce the bosons via a
gauge principle. For convenience, we work a Hamiltonian rather
than a Lagrangian approach, (so that we work over space $X =
{\Bbb R}^3$ rather than over space-time) and in a radiation
gauge.

The real structure of the anticommutation relations, encoded in
a complexified Clifford algebra Cliff$(W,Q)$ of a real
pre-Hilbert space $W$ with inner product $Q$, turns out to be
more fundamental. The Clifford algebra is the unital algebra
generated by elements $\Psi(\xi)$, so that for any
cross-sections $\xi, \eta\in W$, the (equal time)
anticommutator satisfies
$$
[\Psi(\xi),\Psi(\eta)]_+ := \Psi(\xi)\Psi(\eta)+ \Psi(\eta)\Psi(\xi) = 2Q(\xi,\eta)1.
$$
 This algebra has a natural antilinear antimorphism or
*-structure, and a normalised trace $\tau:
{\textrm{Cliff}}(W,Q) \to {\Bbb C}$, \cite{PR}. In QED one can
take $W$ to be the smooth, fast decreasing (more precisely,
Schwartz) cross-sections of a Dirac spinor bundle over ${\Bbb
R}^3$, considered as a real vector space, and $Q(\xi,\eta)$ the
integral over ${\Bbb R}^3$ of the real part of the spinor inner
product $\langle{\xi({\bf x})},{\eta({\bf x})}\rangle$.

For any complex structure $J$ on $W$, we can form the complex
*-algebra generated by elements $\Psi_J(\xi) =
\frac12\left(\Psi(\xi) - i\Psi(J\xi)\right)$, which satisfy the
anticommutation relations
$$
[\Psi_J(\xi)^*,\Psi_J(\eta)]_+ = \langle{\xi},{\eta}\rangle
 = \int_{{\Bbb R}^3} \langle{\xi({\bf x})},{\eta({\bf x})}\rangle\,d^3{\bf x},
$$
(with $\Psi_J$ a creation operator). The procedure can be
reversed by taking self-adjoint operators $\Psi(\xi) =
\Psi_J(\xi) + \Psi_J(\xi)^*$ which satisfy the original
Clifford algebra relations. In the CAR $\xi$ and $\eta$ may be
regarded as elements of the complexification $W_{{\Bbb C}} =
W\otimes{\Bbb C}$. This is the direct sum of two complex
subspaces $W_J^\pm = \ker(1 \pm iJ)$, each of which is easily
checked to be isotropic with respect to the complex bilinear
extension of $Q$. Formally, $\Psi_J(\xi)
=\frac12\left(\Psi(\xi) - i\Psi(J\xi)\right) \sim
\Psi\left(\frac12(1 - iJ)\xi)\right)$ so that we may as well
take  $\xi \in W^+_J$. Similarly in $\Psi_J(\eta)$ we can take
$\eta \in W^-_J$.

We can now define the Fock representation of the CAR associated
with $J$ which takes place on the Hilbert space completion of
the exterior algebra $\bigwedge W_J^+$ equipped with the inner
product derived by extending $Q$. The creation operator
$\Psi_J(\xi)$ acts as exterior multiplication by $\xi \in
W^+_J$. Its adjoint turns out to be an inner multiplication and
the CAR are easily verified, \cite{PR}. This Fock
representation has an obvious cyclic vector $\Phi_J = 1 \in
{\Bbb C} = \bigwedge^0W_J^+$, called the Fock vacuum, which is
annihilated by all the $\Psi_J(\xi)^*$. This provides a
correspondence between complex structures $J$ and vacuum states
$\Phi_J$. Different physical states correspond to different
complex structures and different complex structures usually
give inequivalent Fock representations.


\section{The gauge bosons}

Quantum Electrodynamics was soon followed by nonabelian gauge
theories, and now gauge symmetries are regarded as fundamental.
This insight provides an alternative way to the traditional
introduction of gauge bosons simply by adding new generators
and commutation relations. Let $G$ be the (global) gauge group
(U(1) for QED, and generally a compact connected  Lie group),
which we suppose to have a unitary representation on spinors.
The natural pointwise action of the local gauge group of smooth
maps tending to 1 at $\infty$, ${\mathcal G} =
{\textrm{Map}}_0({\Bbb R}^3,G)$ on the space $W$ of sections of
the spin bundle preserves the inner product $Q$, and, for
$\chi\in {\mathcal G}$ we may define an automorphism
$\alpha_\chi$ of Cliff$(W,Q)$ by $\alpha_\chi[\Psi(\xi)] =
\Psi(\chi.\xi)$. (In QED one has $G=U(1)$ which certainly has a
unitary representation on the complex spinors.) When $G = U(1)$
the abelian gauge group, ${\mathcal G}$, is amenable;
 Carey and Grundling showed that for $G=U(n)$ or $SU(n)$, and also
for smooth subalgebras on compactified ${\Bbb R}^3$  there is
still a topology with respect to which  ${\mathcal G}$ is
amenable, \cite{ALC5}.

This action of the gauge group is not compatible with the Dirac
equation governing the dynamics of the fermions, so we
introduce a connection $\nabla$ on the spinor bundle, which
compensates for the fermionic gauge transformation by changing
to $\chi^{-1}\nabla\chi$ under the gauge group action. In
practice it is easier to replace $\nabla$ by its Lie algebra
$\frak{g}$-valued connection 1-form $\omega \in \Omega =
\Omega^1(X,\frak{g})$, and then the gauge group action is given
by $\omega \mapsto \chi^{-1}\omega\chi +i\chi^{-1}d\chi$. For
an abelian group $G$ such as $U(1)$ this simplifies to $\omega
\mapsto \omega +i\chi^{-1}d\chi$.

The fermion dynamics depend on $\omega$, and for each
connection $\omega$ there is a fermion space $W_\omega$, giving
a  pre-Hilbert bundle $W$ over the space of connections, and an
associated Clifford bundle with fibres ${\textrm{
Cliff}}(W_\omega,Q)$. So fermions should now be regarded as
cross-sections of this bundle. The gauge principle requires
that the change of connection agrees with the automorphism
$\alpha_\chi$. That is for a Clifford algebra valued section
$\Xi$ of the bundle over connections one has
$$
\Xi(\omega+i\chi^{-1}d\chi) = \alpha_\chi[\Xi(\omega)],
$$
precisely the condition defining an induced algebra:
$$
 {\textrm{ind}}_{\mathcal G}^\Omega({\textrm{Cliff}},\alpha) =
\{\Xi\in {\textrm{Cliff}} : \Xi(\omega+i\chi^{-1}d\chi) = \alpha_\chi[\Xi(\omega)]\}.
$$
Similar algebras appear in string theory, where the
infinite-dimensional vector group $\Omega$ is replaced by a
locally compact vector group, and the gauge group ${\mathcal
G}$ by a maximal rank lattice subgroup, so that the quotient
$\Omega/{\mathcal G}$ is a torus, and more generally one
studies principal torus bundles, represented by continuous
trace algebras with non-trivial Dixmier--Douady class (or
$H$-flux), \cite{MR,BHM} Non-trivial classes seem to be
unnecessary for our abelian gauge theories, though we can still
associate the algebra with a $\Omega/{\mathcal G}$-bundle,
\cite {ALC3,ALC4}.

The induced algebra carries a Mackey action of functions on
$\Omega/{\mathcal G}$: any function $F$ on $\Omega/{\mathcal
G}$ lifts to $\widetilde{F}$ on $\Omega$ which has a
multiplication action on the induced algebra,
$$
(F\cdot\Xi)(\omega) = \widetilde{F}(\omega)\Xi(\omega).
$$

In our Hamiltonian picture the connection form $\omega$
represents the magnetic vector potential ${\bf A}$, and we can
identify the multiplication action of $F$ with the action of
$\widetilde{F}(e{\bf A}/\hbar)$, where ${\bf A}$ is the
quantised magnetic potential. By construction this depends only
on the gauge equivalence class of ${\bf A}$.

Since $\Omega$ is not locally compact it is not obvious which
topology or algebra of functions to use. Following the usual
conventions of algebraic quantum field theory, we start with
the C$^*$-algebra generated by elements $\phi_a$ of a dual
group, more precisely, these can be obtained from elements $a$
of a dual vector space $\widehat{\Omega}$ of continuous linear
functionals  $\Omega$, so that
$$
\phi_a(\omega) = e^{ia(\omega)}.
$$
We thus take finite linear combinations of these functions with
pointwise multiplication.  For the sections of the bundle over
the connection space we can similarly use linear combinations
of products of spinors with functions $\phi_a$.

An infinite-dimensional  vector group $\Omega$ lacks a
canonical Pontryagin dual. Many spaces are in duality to
$\Omega$, (for locally convex spaces the Mackey--Arens Theorem
characterises the dual pairs.) Besides an algebraic dual, one
might take the geometric holonomy dual of the Rovelli--Smolin
transform, \cite[Ch.14]{A}, which takes dual elements labelled
by loops $\gamma$ in ${\Bbb R}^3$ paired with potentials by the
holonomy of the loop
$$
\gamma: \omega \to \int_\gamma \omega.
$$
Whichever dual one uses, the functions lifted from
$\Omega/{\mathcal G}$ are given by $a \in {\mathcal G}^\perp$,
that is dual elements which vanish on gauge-trivial
connections. This shows explicitly that a generalised fixed
point algebra exists. (Ashtekar and Lewandowski have shown that
the spectrum of the associated unital C$^*$-algebra is a
compactification of $\Omega/{\mathcal G}$, \cite{AL}. This
could potentially provide another algebra of continuous
functions.)


\section{Electric fields}

Dynamically, magnetic fields oscillate into electric fields, so
the above description is still incomplete. In string theory one
forms the dual crossed product algebra coming from a natural
action of $\Omega$, and then $T$-duality then turns out to be
Takai duality, \cite{T}. We do not need the duality here but
for abelian $G$ we can similarly form a crossed product.

There is an $\Omega$-action on the induced algebra given by
$$
\tau_u[\Xi](\omega) = \Xi(\omega+u),
$$
which one may check to be consistent with the gauge condition.
This action allows us to form the crossed product algebra
$$
{\mathcal A} = {\textrm{ind}}_{{\mathcal G}}^\Omega({\textrm{Cliff}}(W,Q),\alpha)\rtimes_\tau \Omega.
$$
We note that, by definition, $ \alpha_\chi
=\tau_{i\chi^{-1}d\chi}$.

The crossed product is effectively generated by the original
algebra and the group, considered as point measures $\delta_u$
concentrated at $u\in \Omega$, with the covariance property
that $\tau_u$ is implemented by the adjoint (conjugation)
action of $\delta_u$. So overall we take the *-algebra
generated by $\phi_{a,u} = \phi_a\delta_u$ with product
$$
\phi_{a,u}*\phi_{b,v} = (\phi_a\delta_u)*(\phi_b\delta_v) = \phi_a(e^{ib(u)}\phi_b)\delta_u*\delta_v
= e^{ib(u)}\phi_{a+b}\delta_{u+v}.
$$
This is clearly noncommutative since
$$
(\phi_{a,u}*\phi_{b,v}) = e^{ib(u)}\phi_{a+b,u+v} = e^{i[b(u)-a(v)]}\phi_{b,v}*\phi_{a,u}.
$$
The $*$ operation is
$$
\phi_{a,u}^* = e^{ia(u)}\phi_{-a,-u}
$$
so that the generators are unitary:
$$
(\phi_{a,u}^**\phi_{a,u})(\omega,\epsilon) = e^{ia(u)}\phi_{-a,-u}*\phi_{a,u}
= e^{ia(u)}e^{-ia(u)}\phi_{0,0} = 1.
$$
In general, the crossed product consists of complex valued \lq
functions\rq\ on $\Omega\times\Omega$ with product and star
$$
(f*g)(\omega,\epsilon) = \int f(\omega,\epsilon_1)g(\omega+\epsilon_1,\epsilon-\epsilon_1)\,d\epsilon_1,
\qquad
f^*(\omega,\epsilon) = \overline{f(\omega-\epsilon,-\epsilon)}.
$$

The ${\mathcal G}$-fixed algebra ${\mathcal B}$ can be
considered as the algebra generated by $\phi_{a,u}$ with $a\in
{\mathcal G}^\perp$, the subgroup of dual elements which map
${\mathcal G}$ to 1. (This demonstrates the existence of the
generalised fixed point algebra, which is not always obvious.)
Although it is customary to handle the boson algebra in this
way, it does yield an algebra with some unphysical
representations. By endowing the groups with a more subtle
topology and using continuous functions it is possible to get
only the physical, regular, representations, \cite{GN}.

There is an action of $\Omega$ on the algebra of cross-sections
of the Clifford bundle over $\Omega$, and so here too we can
form a crossed product algebra, which will described both the
fermion and boson fields. The Clifford bundle is graded and its
0-component gives rise to the gauge boson algebra just
described
$$ {\mathcal B} =
{\textrm{ind}}_{{\mathcal G}}^\Omega({\Bbb C},\alpha)\rtimes_\tau \Omega.
$$
 (A similar construction occurs in the noncommutative
geometric approach to T-duality in \cite{MR} but with $\Omega$
a finite-dimensional vector group and ${\mathcal G}$ a maximal
rank lattice in $\Omega$ so that $\Omega/{\mathcal G}$ is a
torus. Then ${\textrm{ind}}_{{\mathcal G}}^\Omega({\Bbb
C},\alpha)$ is the C$^*$-algebra representing a principal
$\Omega/{\mathcal G}$-bundle, and ${\textrm{ind}}_{{\mathcal
G}}^\Omega({\Bbb C},\alpha)\rtimes_\tau \Omega$ is the
C$^*$-algebra associated with its T-dual torus bundle. From
this viewpoint T-duality is just Takai duality for particular
C$^*$-algebras.) In the context of this paper the crossed
product is equivalent to the CCR algebra for the bosonic gauge
theory, including longitudinal modes. The full Clifford algebra
is a module for this boson algebra (though no of finite rank),
so that we can consider this as a kind of vector bundle over
the noncommutative space associated with the boson algebra. The
non-commutativity of the bosons means that an uncertainty
principle constrains the fermionic cross sections.

In QED the canonical commutation relations for smeared gauge
boson fields, ${\mathbf A}({\mathbf a})$, ${\mathbf E}({\mathbf
u})$, (${\mathbf a}, {\mathbf u}$ in the Schwartz space
${\mathcal S}(X,{\frak g})$, and ${\mathbf E}({\mathbf u}) =
\int {\mathbf E}({\mathbf x}).{\mathbf u}({\mathbf x})\,d^3{\bf
x}$) are
$$
[{\mathbf E}({\mathbf u}),{\mathbf A}({\mathbf a})] =
-i\frac{\hbar}{\epsilon_0}\int_{{\Bbb R}^3} {\mathbf u}({\mathbf x}).{\mathbf a}^\perp({\mathbf
x})\,d^3{\mathbf x},
$$
where ${\mathbf a}^\perp$ denotes the transverse part.
(Henceforth we shall assume that the gauge has been fixed so
that ${\mathbf a}$ is transverse, and drop the $\perp$.) This
means that the electric field generates translations. More
precisely, exponentiating this to the group
$$
e^{i{\mathbf u}.{\mathbf E}}\left(\frac{e}{\hbar}{\mathbf A}({\mathbf a})\right)e^{-i{\mathbf u}.{\mathbf E}}
 = \left(\frac{e}{\hbar}{\mathbf A}({\mathbf a}) + \frac{e}{\epsilon_0}\int{\mathbf a}.{\mathbf u}\right),
$$
so that $\tau_{eu/\epsilon_0}$ is implemented by
$$
\left[\exp\left(\frac{i\epsilon_0}{e}{\mathbf E}({\mathbf u})\right)\right]
$$
where ${\mathbf E}$ is the quantised electric field.

Our Hamiltonian gauge-fixed description is not manifestly
Lorentz invariant, since Lorentz transformations mix the
magnetic and electric fields, (forcing further gauge
transformations \cite{ALC1}.) This suggests the interesting
question of how the Lorentz symmetry manifests itself in the
induced crossed product.


\section{The Poisson--Gauss law}

Gauge invariance allows us to remove the longitudinal magnetic
fields,  but longitudinal electric fields remain as part of the
translation group.

For $f\in C^\infty_0(X)$ set $\chi_f = \exp(-if)$, so that
$\chi_f^{-1}d\chi_f = -i\,df$, we have
$$
\tau_{df}[\xi](\omega) = \xi(\omega+ df) = \xi(\omega+i\chi_f^{-1}d\chi_f) =
\chi_f\cdot\xi(\omega) = e^{-if}\xi(\omega) .
$$
By definition,$\tau_{edf/\epsilon_0}$ is implemented in the
crossed product by $\exp(i{\mathbf E}\cdot(\nabla f))$, so that
implements multiplication by $\exp(-ief/\epsilon_0)$.

On the other hand a formal calculation shows the implementor
should be $\exp(-i\rho(f)/\epsilon_0)$, where the charge
density operator
$$
\rho(f) = \int_X f({\mathbf x})\,e\Psi_J(x)^*\Psi_J(x)\,d^3{\mathbf x}.
$$
(We need to work with the complex algebra here to incorporate
complex gauge factors.) So (differentiating and using the
Divergence Theorem): we get the Poisson--Gauss Law:
$$
(\nabla\cdot{\mathbf E})(f) = -{\mathbf E}\cdot(\nabla f) =
\frac{1}{\epsilon_0}\rho(f).
$$

More directly, we have the differentiated version
$$
{\textrm{ad}}\left((\nabla f)\cdot{\mathbf E}\right)[\Psi_J(\xi)]
= \frac{e}{\epsilon_0}\Psi_J(f\xi).
$$
Using ${\mathbf u} = \nabla f$, or $f =
\nabla\cdot\nabla^{-2}{\mathbf u}$,
$$
({\mathbf u}\cdot{\mathbf E})\Psi_J(\xi)
 = \Psi_J(\xi)({\mathbf e}\cdot{\mathbf E})+ \frac{e}{\epsilon_0}\Psi_J(\nabla\cdot\nabla^{-2}{\mathbf u}\xi),
 $$
and, using the fact that $\nabla^{-2}$ is an integral operator
with integral kernel $1/(4\pi|{\mathbf x} - \mathbf{y}|)$, and
working with unsmeared fields, we obtain
$$
{\mathbf E}({\mathbf x})\Psi_J({\mathbf y})
 = \Psi_J({\mathbf y})\left[{\mathbf E}({\mathbf x})+ \nabla\frac{e}{4\pi\epsilon_0|{\mathbf x} - \mathbf{y}|}\right].
$$
This can be interpreted as saying that creating a fermion using
$\Psi({\mathbf y})$, also creates its Coulomb field. Similar
ideas appear, without the framework of induced algebras, in
\cite[\S\S79-80]{D}.


\section{Rieffel inducing}


Despite their advantage of being explicit, the above procedures
do not easily extend to non-abelian gauge theories. (Apart from
the obvious difficulty that the $\Omega$ action only preserves
the inducing constraint in the abelian case, there can be
obstructions to extensions in the non-abelian case,
\cite{ALC3,ALC4}.)

In extending the approach it is useful to work with Rieffel's
bimodule inducing, which like Mackey's construction,  allows
one to induce modules as well as algebras. Initially we shall
just do this in an algebraic setting, ignoring the Hilbert
structure. Let ${\mathcal B}$ and ${\mathcal C}$ be algebras,
and ${\mathcal E}$ a ${\mathcal B}$-${\mathcal C}$- bimodule,
(i.e. a left ${\mathcal B}$, right ${\mathcal C}$-module, with
commuting actions.) Then, from from a left ${\mathcal
C}$-module $M$ one can induce a left ${\mathcal B}$-module
${\mathcal E}\otimes_{{\mathcal C}}M$, (the quotient of
${\mathcal E}\otimes M$ by the subspace generated by $\{
e.c\otimes m - e\otimes c.m: e\in {\mathcal E}, c\in {\mathcal
C}, m\in M\}$). The ${\mathcal B}$-action is given by
$b.(e\otimes_{\mathcal C} m) = (b.e)\otimes_{\mathcal C} m$.

When ${\mathcal C}$ is a group algebra the quotient ${\mathcal
E}\otimes_{\mathcal C} M$, can be expressed more simply, by
regarding ${\mathcal E}$ and $M$ as having a group action. For
$h$ in the group we require $e.h\otimes m = e\otimes h.m$ in
the quotient, or, equivalently, $e\otimes m = e.h^{-1}\otimes
h.m$. This is just the requirement that we are in the fixed
point subspace under the action $h: e\otimes m  \mapsto
e.h^{-1} \otimes h.m$.

Omitting the technical details, Rieffel's method extends this
algebraic theory to C$^*$-algebras ${\mathcal B}$ and
${\mathcal C}$ by assuming that the bimodule ${\mathcal E}$ has
a ${\mathcal C}$-valued inner product, ${\mathcal
E}\times{\mathcal E} \to {\mathcal C}$, which is ${\mathcal
C}$-linear in the second variable
$$
\langle e_1, e_2.c\rangle = \langle e_1, e_2\rangle .c,
$$
and *-symmetric
$$
\langle e_1, e_2\rangle^* = \langle e_2, e_1\rangle,
$$
as well as positive $\langle e, e\rangle \geq 0$, that is
positive in the C$^*$-algebra ${\mathcal C}$. When the
${\mathcal C}$-module $M$ has a Hilbert space structure
consistent with the C$^*$-algebra structure of ${\mathcal C}$
we can endow ${\mathcal E}\otimes M$ with the inner product
$$
\langle e_1\otimes m_1, e_2\otimes m_2\rangle =
\langle m_1, \langle e_1, e_2\rangle.m_2\rangle.
$$
In general this is not positive definite so, to get a Hilbert
space we need to factor out by the radical (the vectors
orthogonal to everything). We now note that
$$
\langle e_1, e_2.c\rangle.m_2 =   \langle e_1, e_2\rangle.(c.m_2),
$$
so  that $(e_2.c) \otimes m_2 - e_2\otimes (c.m_2)$ is always
in the radical. This means that the quotient by the radical
will be a quotient of ${\mathcal E}\otimes_{\mathcal C} M$,
showing the connection with  the algebraic approach. (Under
suitable assumpions the inner product is positive definite on
${\mathcal E}\otimes_{\mathcal C} M$, and we shall assume this
to be the case.)

We shall take ${\mathcal C}$ to be a convolution algebra of
functions on ${\mathcal G}$, and ${\mathcal B}$ to be the boson
algebra consistent with our previous notation. It turns out
that an appropriate bimodule is the algebra ${\mathcal E}$ for
the boson fields when gauge symmetries are ignored. (For
example, in the abelian case we would take
$$
{\mathcal E} = {\textrm{ind}}_{\{1\}}^\Omega({\textrm
{Cliff}}(W,Q),\alpha)\rtimes_\tau \Omega \cong
C(\Omega,{\textrm{Cliff}}(W,Q) \rtimes_\tau \Omega,
$$
 instead of ${\textrm
{ind}}_{{\mathcal
G}}^\Omega({\textrm{Cliff}}(W,Q),\alpha)\rtimes_\tau \Omega$,
but, in general, we could use any well-defined *-algebra
${\mathcal E}$ for the bosons ignoring gauge symmetries, along
with a compatible space-time description of the fermions.)   We
assume that there is a ${\mathcal G}$-action $\alpha$ on
${\mathcal E}$ as before, and  then the generalised fixed point
algebra (in the multiplier algebra of ${\mathcal E}$) is
${\mathcal B} = {\mathcal E}^{\mathcal G}$) which has a
multiplication action on ${\mathcal E}$ commuting with the
action of ${\mathcal G}$, so that we may use Rieffel inducing
to induce ${\mathcal G}$-modules to modules for ${\mathcal B}
$.
 (The amenability of ${\mathcal G}$ enables one to prove the existence of
generalised fixed point subspaces, although that can also be
done explicitly. See also \cite{ALC}.) The algebra ${\mathcal
E}$ is a group algebra of a nilpotent group, a central
extension of the vector group $\Omega\times\Omega$, and it is
therefore amenable. A pre-inner product $\langle
e_1|e_2\rangle$ on ${\mathcal E}$ can be constructed by taking
the invariant mean of the product $e_1^*e_2$, and we obtain the
${\mathcal C}$-valued inner product by defining
$$
\langle e_1, e_2\rangle(\chi) = \langle e_1, \chi.e_2\rangle,
$$
which turns out to have the correct properties, provided that
we specify the gauge group algebra ${\mathcal C}$ to include
these functions.

Since we have a group algebra The inducing procedure takes a
${\mathcal G}$-module $M$ to $F(M) = {\mathcal
E}\otimes_{\mathcal G} M = ({\mathcal E}\otimes M)^{\mathcal
G}$, where the tensor product action of $\chi\in{\mathcal G}$
sends $\mu\otimes m$ to $\alpha_\chi[\mu]\otimes \chi.m$.
(Strictly speaking, we should induce from the group algebra,
but this is equivalent, and simpler.)

For any ${\mathcal G}$-intertwiner $f$ between ${\mathcal
G}$-modules $M \to N$, we can define  $F(f):\mu\otimes m
\mapsto \mu\otimes f(m)$, which commutes with the ${\mathcal
G}$ action, and so preserves the gauge fixed algebra. The
action of $F(f)$ on the second tensor factor commutes with the
actions of ${\mathcal B}$ on the first, so that $F(f)$ is a
${\mathcal B}$-morphism and $F$ defines a functor.

\medskip\noindent
{\bf Example 1.} Consider the case of $M= {\Bbb C}$ the trivial
${\mathcal G}$ module. The induced module can be determined by
using the generator $1\in {\Bbb C}$, which enables one to
identify ${\mathcal E}\otimes{\Bbb C}$ with ${\mathcal E}$ by
$\mu\otimes1 \mapsto \mu$. Under this identification
$({\mathcal E}\otimes{\Bbb C})^{\mathcal G} = {\mathcal
E}^{\mathcal G}$, so that $F({\Bbb C})$ is the fixed point
algebra ${\mathcal E}^{\mathcal G} = {\mathcal B}$. The algebra
product is inherited from that on ${\mathcal E}$ and the normal
product on ${\Bbb C}$.

\begin{theorem}
The map $F: M \mapsto {\mathcal E}\otimes_{\mathcal G} M$ from
${\mathcal G}$-modules to ${\mathcal B}$--bimodules, which
takes a ${\mathcal G}$-intertwining operator $f: M\to N$  to
$F(f):\mu\otimes_{\mathcal G} m \mapsto \mu\otimes_{\mathcal G}
f(m)$ defines a functor from ${\mathcal G}$-modules and
intertwiners to ${\mathcal B}$-bimodules and intertwiners.
\end{theorem}

\begin{proof}
We have proved most of this except the statements about
${\mathcal B}$-bimodules. Since ${\mathcal G}$ acts by
automorphisms of ${\mathcal E}$, we have, denoting the product
of $\mu, \nu \in {\mathcal E}$ by $\mu*\nu$,
$\alpha_\chi[\mu*\nu] = \alpha_\chi[\mu]*\alpha_\chi[\nu]$.
When $\mu \in {\mathcal E}^{\mathcal G}$ this gives
$$
\alpha_\chi[\mu*\nu] \otimes \chi.n =
\mu*\alpha_\chi[\nu]\otimes \chi.n = (\mu\otimes 1)*\chi(\nu\otimes n),
$$
so that we can define an action of ${\mathcal B}= {\mathcal
E}^{\mathcal G}$ on $({\mathcal E}\otimes N)^{\mathcal G}$ by
$\mu*[\nu\otimes n] = (\mu*\nu)\otimes n$. We could equally
well have used right multiplication by $f$, so that $F(M)$ is a
${\mathcal B}$-bimodule action. The action commutes with the
${\mathcal G}$ intertwining operators, which affect the other
factor in the tensor product.
\end{proof}


\section{A monoidal functor}

 A strict monoidal or tensor category is a category ${\mathcal C}$, together
with (i) an associative bifunctor $\otimes: {\mathcal
C}\times{\mathcal C} \to {\mathcal C}$, and (ii) a unit object
$U$ such that $U\otimes A \cong A \cong A\otimes U$ for all
objects $A$, satisfying the obvious consistency conditions that
the isomorphisms agree for $U\otimes U \cong U$, and for
$$
A\otimes B \cong (A\otimes U)\otimes B \cong A\otimes (U\otimes B) \cong A\otimes B.
$$

\medskip
\noindent{\bf Example 2.}
 The ${\mathcal G}$-modules form a monoidal category under the tensor
product (with the tensor product action, $\chi\otimes\chi$, of
${\mathcal G}$), and with unit the trivial module ${\Bbb C}$.
${\mathcal B}$-bimodules form a monoidal category with tensor
product $\otimes_{\mathcal B}$ and unit object ${\mathcal B}$.

\medskip
The question now arises as to whether we can interpret $F: M
\mapsto ({\mathcal E}\otimes M)^{\mathcal G}$ as a monoidal
functor. If so then the unit object should be $F({\Bbb C})
\cong {\mathcal B}$, suggesting that $F$ maps from the monoidal
category of ${\mathcal G}$-modules to the monoidal category of
${\mathcal B}$-bimodules. The main extra piece of information
needed is a map $F_{MN}$, for each pair of ${\mathcal
G}$-modules $M$ and $N$, which takes $F(M)\otimes_{\mathcal B}
F(N)$ to $F(M\otimes N)$. The obvious map is to start with
$$
(\mu\otimes_{\mathcal G}  m)\otimes (\nu\otimes_{\mathcal G}  n) \mapsto (\mu*\nu)\otimes_{\mathcal G} (m\otimes n).
$$
For $\beta\in {\mathcal B}$, $((\mu*\beta)\otimes_{\mathcal G}
m)\otimes(\nu\otimes_{\mathcal G} n) - (\mu\otimes_{\mathcal G}
m)\otimes (\beta*\nu\otimes_{\mathcal G} n)$ maps to
$$
[((\mu*\beta)*\nu) - (\mu*(\beta*\nu))]\otimes_{\mathcal G}
(m\otimes n) = 0,
$$
so that the right hand side depends only on
$(\mu\otimes_{\mathcal G} m)\otimes_{\mathcal B}
(\nu\otimes_{\mathcal G}  n)$, and we can regard $F_{MN}$ as a
morphism $F(M)\otimes_{\mathcal B} F(N)$ to $F(M\otimes N)$.

\begin{theorem}
The functor $F$, together with
$$
F_{MN}:(\mu\otimes_{\mathcal G}  m)\otimes_{\mathcal B}(\nu\otimes_{\mathcal G}  n)
\mapsto (\mu*\nu)\otimes_{\mathcal G} (m\otimes n)
$$
and  the identification ${\mathcal B} \to F({\Bbb C})$, already
used, is a monoidal functor, from ${\mathcal G}$-modules with
the normal tensor product to ${\mathcal B}$-bimodules with the
tensor product $\otimes_{\mathcal B}$.
\end{theorem}

\begin{proof}
Due to the associativity of the convolution multiplication we
see both that this is consistent with strict associativity
$\xi\otimes(\eta\otimes\zeta) \to
(\xi\otimes\eta)\otimes\zeta$. We have checked out that the
left-hand side makes sense, but we also have
\begin{eqnarray*}
\chi.[(\mu\otimes m)\otimes\chi.[(\nu\otimes n)]
&=& ((\alpha_\chi[\mu]\otimes \chi.m)\otimes(\alpha_\chi[\nu]\otimes \chi.n)\\
&\mapsto&
(\alpha_\chi[\mu*\nu]\otimes \chi.(m\otimes n)\\
&=& \chi.[(\mu*\nu)\otimes (m\otimes n)],
\end{eqnarray*}
so that products map ${\mathcal G}$-fixed elements to
${\mathcal G}$-fixed elements.
\end{proof}

In the abelian case, we see that
\begin{eqnarray*}
(\phi_{a,u}\otimes m)\otimes (\phi_{b,v}\otimes n)
&\mapsto& (\phi_{a,u}*\phi_{b,v})\otimes(m\otimes n)\\
&=& e^{ib(u)}\phi_{a+b,u+v}\otimes(m\otimes n),
\end{eqnarray*}
so that $F_{MN}$ sets up an isomorphism. This means that we
actually have a strong monoidal functor.

The noncommutativity of ${\mathcal E}$ leads to a braiding in
the image category.

\begin{theorem}
When $F$ is a strong monoidal functor, the tensor category of
${\mathcal B}$-bimodules with tensor product $\otimes_{\mathcal
B}$ is a braided category with symmetric braiding $\Phi_F =
F_{NM}^{-1}\circ F(\Phi)\circ F_{MN}: F(M)\otimes_{\mathcal B}
F(N) \to F(N)\otimes_{\mathcal B} F(M)$, as in the diagram
$$
F(M)\otimes_{\mathcal B} F(N) \to F(M\otimes N) \to F(N\otimes M) \cong F(N)\otimes_{\mathcal B} F(M),
$$
where the outer maps are given by the consistency maps $F_{MN}$
and $F_{NM}^{-1}$ and the middle map is $F(\Phi)$.
\end{theorem}

\begin{proof}
The tensor product of ${\mathcal G}$-modules is braided
trivially by the flip ${\mathcal G}$-morphism $\Phi: m\otimes
n\mapsto n\otimes m$. Thus
$$
F(M)\otimes_{\mathcal B} F(N) \to F(M\otimes N) \to F(N\otimes M) \cong F(N)\otimes_{\mathcal B} F(M),
$$
where the outer maps are given by the consistency maps $F_{MN}$
and $F_{NM}^{-1}$ and the middle map is $F(\Phi)$. Although
this braiding is non-trivial it is symmetric, since $\Phi^2 =
{\textrm{id}}$ gives $F(\Phi)^2 = {\textrm{id}}$. For many
purposes this is almost as good as being the standard flip
braiding.
\end{proof}

We can see that the braiding is non-trivial in the abelian case
by an explicit calculation:
\begin{eqnarray*}
(\phi_{a,u}\otimes m)\otimes (\phi_{b,v}\otimes n) &\mapsto&
e^{ib(u)}\phi_{a+b,u+v}\otimes(m\otimes n)\\
(\phi_{b,v}\otimes n)\otimes(\phi_{a,e}\otimes m) &\mapsto&
e^{ia(v)}\phi_{a+b,u+v}\otimes(n\otimes m).
\end{eqnarray*}
We can also easily check the symmetry in this case.


\section{Transferring fermionic structure to the whole QED
algebra}

Using the functor $F$, any structures which can be defined
categorically for ${\mathcal G}$-modules can now be defined for
${\mathcal B}$-modules.

The argument which gave the braiding similarly leads to the
following result.

\begin{lemma}
For any morphism of ${\mathcal G}$-modules $\phi:M\otimes N \to
P$ there is a morphism of ${\mathcal B}$-bimodules $\phi_F:
F(M)\otimes_{\mathcal B} F(N) \to F(P)$, defined by $\phi_F =
F(\phi)\circ F_{MN}$.
\end{lemma}

This has many useful corollaries,  such as the following.

\begin{corollary}
The gauge group ${\mathcal G}$ acts as automorphisms of an
algebra ${\mathcal A}$ if and only if its multiplication $\mu$
is a ${\mathcal G}$-morphism. In this case $F({\mathcal A})$ is
an algebra with multiplication $F(\mu)\circ F_{{\mathcal
A}{\mathcal A}}$.
\end{corollary}

\begin{proof}
The multiplication map $\mu:{\mathcal A}\otimes{\mathcal A} \to
{\mathcal A}$  admits ${\mathcal G}$ as algebra automorphims if
and only if $\mu$ intertwines the actions of ${\mathcal G}$ on
${\mathcal A}\otimes{\mathcal A}$ and ${\mathcal A}$, so that
$\mu\circ(\chi\otimes\chi) = \chi\circ\mu$, which is precisely
the condition that $\mu$ be a morphism in the category. Under
this condition, take $M=N=P = {\mathcal A}$.
\end{proof}

There is a similar argument for modules.

\begin{corollary}
Let ${\mathcal A}$ be an algebra on which ${\mathcal G}$ acts
as automorphisms, and $N$ a covariant $({\mathcal A},{\mathcal
G})$-module defined by an action $a: {\mathcal A}\otimes N \to
N$, on which ${\mathcal G}$ also acts in a covariant way, that
is $\chi A\chi^{-1} = \chi[A]$ for all $A\in {\mathcal A}$.
Then $F(N)$ is an $F({\mathcal A})$-module.
\end{corollary}

\begin{proof}
Take $M = {\mathcal A}$, and $P=N$.
\end{proof}

\begin{corollary}
Let $W$ be a ${\mathcal G}$-module with a ${\mathcal
G}$-invariant quadratic form $Q:W\otimes W \to {\Bbb C}$. Then
$Q_F: F(W)\otimes_{\mathcal B} F(W) \to F({\Bbb C}) = {\mathcal
B}$ is a ${\mathcal B}$-valued quadratic form on $F(W)$.
\end{corollary}

\begin{proof}
Take $M=N=W$ and  $P = {\Bbb C}$.
\end{proof}

In general, the functor $F$ takes any structure which can be
defined in the category of ${\mathcal G}$ modules, to a similar
structure in the new monoidal category.

 The Clifford algebra ${\textrm{Cliff}}(W,Q)$ can be
defined as the universal unital complex algebra for maps $f$
from $W$ to an algebra ${\mathcal C}$ such that one has a
commutative diagram
\begin{center}
\begin{picture}(200,90)(10,-10)
 \put(40,70){\makebox(0,0){$W\otimes W$}}
 \put(60,70){\vector(1,0){95}}
 \put(20,35){\makebox(0,0){id$\,+\Phi$}}
 \put(100,75){\makebox(0,0){$2Q$}}
 \put(160,70){\makebox(0,0){${\Bbb C}$}}
 \put(40,60){\vector(0,-1){55}}
 \put(40,0){\makebox(0,0){$W\otimes W$}}
 \put(60,0){\vector(1,0){25}}
 \put(75,7){\makebox(0,0){$f\otimes f$}}
 \put(105,0){\makebox(0,0){${\mathcal C}\otimes{\mathcal C}$}}
 \put(125,0){\vector(1,0){25}}
 \put(160,0){\makebox(0,0){${\mathcal C}$}}
 \put(135,7){\makebox(0,0){$m$}}
 \put(160,60){\vector(0,-1){55}}
 \put(170,35){\makebox(0,0){$\times 1$}}
\end{picture}
\end{center}
and for any such algebra $f$ there is a morphism $f_*:{\rm
Cliff}(W,Q) \to {\mathcal C}$ whose composition with $W \to
{\rm Cliff}(W,Q)$ is $f$.

\begin{theorem}
The algebra $F({\textrm{Cliff}}(W,Q))$ is a universal object
for the corresponding diagrams in the braided category of
${\mathcal B}$-bimodules, and so can be regarded as a Clifford
algebra ${\textrm{Cliff}}(F(W),F(Q))$ in that category.
\end{theorem}

This means that the algebra for interacting quantum
electrodynamics can be regarded as a Clifford algebra over the
gauge boson algebra, and inherits interesting features, coming
from the trace, and antilinear anti-automorphism \cite{PR}.

\begin{corollary}
The fermion Clifford algebra has a unique braided commutative
normalised conditional expectation $F(\tau): {\rm
Cliff}(F(W),F(Q)) \to {\mathcal B}$, which is ${\mathcal
B}$-linear.
\end{corollary}

\begin{proof}
There is a unique (and therefore ${\mathcal G}$-invariant)
normalised trace $\tau: {\textrm{Cliff}}(W,Q) \to {\Bbb C}$,
and this gives a map $F(\tau): {\textrm{Cliff}}(F(W),F(Q)) \to
{\mathcal B}$. Linearity of the original trace is expressed by
the commutativity of the diagram
\begin{center}
\begin{picture}(300,70)(20,-10)
 \put(50,50){\makebox(0,0){${\Bbb C}\otimes A$}}
 \put(70,50){\vector(1,0){60}}
 \put(100,55){\makebox(0,0){${\textrm{id}}\otimes\tau$}}
 \put(150,50){\makebox(0,0){${\Bbb C}\otimes {\Bbb C}$}}
 \put(50,45){\vector(0,-1){40}}
 \put(35,25){\makebox(0,0){mult}}
 \put(50,0){\makebox(0,0){$A$}}
 \put(70,0){\vector(1,0){60}}
 \put(150,0){\makebox(0,0){${\Bbb C}$}}
 \put(150,45){\vector(0,-1){40}}
 \put(165,25){\makebox(0,0){mult}}
 \put(110,-5){\makebox(0,0){$\tau$}}
 \put(250,50){\makebox(0,0){$A\otimes{\Bbb C}$}}
 \put(230,50){\vector(-1,0){60}}
 \put(250,45){\vector(0,-1){40}}
 \put(285,25){\makebox(0,0){mult}}
 \put(250,0){\makebox(0,0){$A$}}
 \put(230,0){\vector(-1,0){60}}
\end{picture}
\end{center}
and application of $F$ shows that $F(\tau)$ is ${\mathcal
B}$-linear (on both sides). Similarly the trace property
$\tau\circ{\textrm{mult}}\circ\Phi = \tau\circ{\textrm{mult}}$
gives $F(\tau)\circ {\textrm{mult}}\circ\Phi_F = F(\tau)\circ
{\textrm{mult}}$, showing that $F(\tau)$ is braided symmetric.
Overall $F(\tau)$ defines a braided symmetric conditional
expectation from the full QED algebra to its bosonic part.
Explicitly we have
 $F(\tau)= {\textrm{id}} \otimes \tau$.
\end{proof}

This shows that, within the new category of ${\mathcal
B}$-${\mathcal B}$-bimodules, the type III QED algebra inherits
some of the type II$_1$ properties of the original Clifford
algebra.


\section{Bar/monoidal dagger categories}

There is yet further structure in these categories. In order to
be able to talk about antilinear operations such as a
sesquilinear inner product on a Hilbert space or the
$*$-structure on a C$^*$-algebra it is useful to work in a bar
or monoidal dagger category \cite{BM,AC,S}, which were devised
for precisely this purpose.  Bar categories are slightly more
convenient for our purposes, as \cite{BM} already contains
several examples of interest, so we shall use them with a
change of notation.

A bar category has a functor from the category to its opposite,
so that an object {\sf bar}:$A\mapsto \overline{A}$, with (i) a
natural equivalence between the identity and {\sf
bar}$\circ${\sf bar} functors; (which we shall actually assume
a strong bar category, and identify $\overline{\overline{A}} =
A$) (ii) a natural morphism $U \mapsto \overline{U}$ from the
unit object (which we shall just write as an identification):
(iii) a natural equivalence $\overline{(A\otimes B)}\to
\overline{B}\otimes\overline{A}$ and consistency with the
associator morphisms. (We have abbreviated the conditions
somewhat, the full definition is in \cite{BM}.)


There is a natural functor on the category of ${\mathcal
G}$-modules which takes a module $M$ to its conjugate
$\overline{M}$ with the conjugate scalar multiplication by
${\Bbb C}$ and action of ${\mathcal G}$. A {\it star object}
$M$ is one where there is an isomorphism $M \to \overline{M}$.

A $*$-algebra is a star object with the isomorphism $\mu
\mapsto \mu^*$ from $M$ to $\overline{M}$.  In particular,
${\mathcal E}$ and the fermionic Clifford algebra are star
objects in the category of ${\mathcal G}$-modules. Moreover,
There is also a bar structure on the ${\mathcal B}$-bimodules,
\cite[Example 2.3]{BM} , and exploiting this with the braiding
this means that we have isomorphisms
$$
{\mathcal E}\otimes\overline{M} \to \overline{M}\otimes{\mathcal E}
\to \overline{M}\otimes\overline{{\mathcal E}} \to \overline{{\mathcal E}\otimes M},
$$
or $(\mu\otimes m)^* = \mu^*\otimes m^*$. From this it follows
that the fixed point sets agree $({\mathcal
E}\otimes\overline{M})^{\mathcal G} \cong \overline{({\mathcal
E}\otimes M)}^{\mathcal G}$, and $F(\overline{M}) \cong
\overline{F(M)}$.
 (This can be seen by regarding the fixed points in $N$ as
 labelling the ${\mathcal G}$-morphisms ${\Bbb C} \to N$. Applying
 ${\sf bar}$  one has ${\Bbb C} \cong \overline{{\Bbb C}} \to
 \overline{N}$, which labels the fixed points of $\overline{N}$.
 Strictly we should have mapped into the multiplier algebra,
 but that is defined by a universal property for algebras
 having $N$ as a two-sided ideal.)

Putting all this together proves the following theorem:

\begin{theorem}
The functor $F$ is a bar functor, that is $F(\overline{M})
\cong \overline{F(M)}$.
\end{theorem}

Since $\tau[m^* m] \geq 0$ for the map $m \mapsto m^*$ on ${\rm
Cliff}(W,Q) = \overline{{\textrm{Cliff}}(W,Q)}$, and
$$
F(\tau)[(\mu\otimes m)^*(\mu\otimes m)]
= (\mu^*\mu)\tau[m^*m],
$$
the trace $F(\tau)$ is positive, This means that we can form
states $\phi\circ F(\tau)$ of $F({\textrm{Cliff}})$ by
composing states $\phi$ of ${\mathcal B}$ with $F(\tau)$.
Rieffel's imprimitivity theorem gives necessary and sufficient
conditions for a given state to be of this form.

We note that whenever there is an inner product with values in
the $*$-algebra ${\mathcal C}$,  $\overline{M}\otimes M \to
{\mathcal C}$ we have a map $\overline{F(M)}\otimes F(M) \to
F({\mathcal C})$. This gives an $F({\mathcal C})$-valued inner
product on $F(M)$, provided that it is positive.


 \section{Carey's Theorem}

The link between complex structures and fermionic vacuum states
mentioned in Section 2 was clarified by Shale and Stinespring
in\cite{ShS}  where they found necessary and sufficient
conditions for two complex structures to define the equivalent
representations of the Clifford algebra. An orthogonal
transformation $T$ of the underlying inner product space $W$
gives rise to a Bogoliubov automorphism of the Clifford
algebra: $\Psi(\xi) \mapsto \Psi(T\xi)$, and the
Shale--Stinespring Theorem also gave a criterion for this to be
implemented by a unitary transformation of the representation
space. In \cite{ALC2} Carey generalised this (and work of
Blattner) to cover quasifree representations of the Clifford
algebra which are completely determined by their two-point
correlation functions. Robinson, \cite{R}, recast the standard
doubling construction used by Carey to show that every
quasi-free representation of a Clifford algebra is obtained as
a restriction of the regular representations on the Hilbert
space
$$
{\mathcal H}_\tau = \{x\in {\textrm{Cliff}}(W,Q)-{\Bbb C}: \tau(x^*x) <\infty\}.
$$

Writing $\lambda$ and $\rho$ for the left and right regular
representations, and letting $\Gamma$ be the implementor of the
orthogonal transformation $-1$, we define:
$$
\varpi(x\oplus y) = \lambda(x) + i\rho(y)\Gamma
$$
to obtain a Fock representation of ${\textrm{Cliff}}(W\oplus W,
Q\oplus Q)$ on ${\mathcal H}_\tau$ with vacuum vector the unit,
1. (Furthermore, the Tomita antiunitary operator coincides with
the canonical conjugation on the Clifford algebra.)

Robinson's construction used only the trace and *, from
Clifford algebra theory, and those natural ingredients are now
available for the generalised Clifford algebra ${\rm
Cliff}(F(W),F(Q))$ too. We hope to discuss their application to
Carey's Theorem in a subsequent paper.


 \section{Conclusions}

The main conclusion of this paper is that algebra of fermions
and gauge bosons can be regarded as a braided Clifford algebra
${\textrm{Cliff}}(F(W),Q_F)$ over the braided commutative
bosonic algebra ${\mathcal B}$, and that it shares many
features with ordinary complex Clifford algebras. This enables
one to incorporate the minimally coupled bosons as well as
fermions, whilst retaining  the spirit of the treatment of free
fermions or those in classical external gauge fields, in, for
example, \cite{ALC2, ALC3, SW, W}. Nonetheless there are
serious differences as soon as one studies the interacting
fermion-boson dynamics.  In a future paper we shall discuss
applications of these ideas to interacting QED.


\end{document}